\documentclass[authoryear,11pt]{elsarticle}
\journal{Statistics and Probability Letters}

\usepackage{setspace}
\usepackage{threeparttable}
\usepackage{amsmath,amssymb,amsthm,mathrsfs}
\usepackage{multirow}
\usepackage{booktabs}
\usepackage{rotating}
\usepackage{color}
\usepackage{bm}
\usepackage{graphicx}
\usepackage{subfigure}
\newcommand{\ve}[1]{\bm{{#1}}}
\newcommand{\vesub}[2]{\bm{{#1}}_{#2}}
\newcommand{\vesup}[2]{\bm{{#1}}^{#2}}
\newcommand{\vess}[3]{\bm{{#1}}_{#2}^{#3}}

\newcommand{\hvesub}[2]{\hat{\ve{#1}}_{#2}}

\newcommand{\tvesub}[2]{\tilde{\ve{#1}}_{#2}}

\newcommand{\tvess}[3]{\tilde{\ve{#1}}_{#2}^{#3}}
\numberwithin{equation}{section}
\begin{document}
	
\begin{frontmatter}
\title{A robust estimation for the extended t-process regression model}
		
\author[a]{Zhanfeng Wang\corref{cor1}}
\author[a]{Kai Li}
\author[b]{Jian Qing Shi}
\address[a]{Department of Statistics and Finance, Management School, University of Science and Technology of China, Hefei, China.}
\address[b]{School of Mathematics and Statistics, Newcastle University, Newcastle, UK.}
\cortext[cor1]{Corresponding author. Email: zfw@ustc.edu.cn.}

\begin{abstract}
	
   Robust estimation and variable selection procedure are developed for the extended t-process regression model with functional data.
	Statistical properties such as consistency of estimators  and predictions are obtained. Numerical studies show that the proposed method performs well.
	
\end{abstract}

\begin{keyword}
	Functional data  \sep Maximum a posterior  \sep Spike and slab priors \sep Information consistency
\end{keyword}
\end{frontmatter}


\section{Introduction}

For functional response variable and functional covariates, this paper considers a concurrent functional regression model
\begin{equation}\label{conmodel}
y_{i}(t_{ij})=f_{i}(\vesub{x}{i}(t_{ij}))+\varepsilon_{i}(t_{ij}), j=1,...,n_i,~i=1,...,m,
\end{equation}	
where $t_{ij}$ is a grid point  which could be temporal or spatial, $f_{i}(\cdot)$ is an unknown function, $\vesub{x}{i}(\cdot)$ is a vector of observed covariates with dimension $p$ and $\varepsilon_{i}(\cdot)$ is an error function.  Hereafter, let $y_{ij}=y_i(t_{ij})$, $\vesub{x}{ij}=\vesub{x}{i}(t_{ij})$ and $\varepsilon_{ij}=\varepsilon_{i}(t_{ij})$.
To estimate $f_{i}(\cdot)$, a process regression model is defined as follows. 
\begin{equation}\label{prmodel}
y_{i}(\ve{x})=f_{i}(\ve{x})+\varepsilon_{i}(\ve{x}), i=1,...,m,
\end{equation}	
where $f_{i}(\cdot)$ and $\varepsilon_{i}(\cdot)$ are assumed to have some stochastic process priors.

Model (\ref{prmodel}) becomes the popular Gaussian process regression (GPR) model, when $\varepsilon_{i}(\cdot)$  and prior for $f_{i}(\cdot)$ have independent Gaussian processes. GPR models are well studied in literature, details can refer to \cite{Rasmussen2006Gaussian}, \cite{Shi2011Gaussian} and therein references.
However, GPR does not give a robust estimation against outliers in the response space.  \cite{Wang2017extend} used an extended t-process (ETP) to build a robust functional regression model, called the extended t-process regression (eTPR) model, where they assumed that $f_{i}(\cdot)$ and $\varepsilon_{i}(\cdot)$ in model (\ref{prmodel}) have a joint extended t-process.
The eTPR model  inherits some nice features from GPR, e.g. the prediction has an analytical form and thus it can be implemented  efficiently; but it also encountered some undesirable problems, for example, the degree of freedom (DF)
involved in ETP is usually over-estimated if a likelihood method is used \citep[see the discussion in e.g.][]{Fernandez1999multi, Lange1993normal, Liu1994statistical}.
Size of DF is crucial to determine the robustness of the model. The eTPR with a smaller DF tends to be more robust against outliers. 
Thence, when the MLE tends to be large, the eTPR model loses robustness.
Actually, 
\cite{Wang2017extend}  stated that eTPR becomes GPR when the degree of freedom tends to  infinity.

This paper proposes a Bayesian approach to estimate the DF.
The proposed procedure has several advantages:
(a)
The DF is estimated via maximum a posterior (MAP) with some suitable priors, resulting in a better and stable estimation. such that estimation approach still has robustness. 
(b) A variable selection procedure is constructed,  by using the spike and slab priors, to parameters involved in covariance functions. This can simplify the covariance structure and improve the accuracy. 
Statistical properties, such as consistency of the MAP  and the 
information consistency of the predicted function, are also investigated.
Numerical studies including simulation results and real examples are presented to show the performance of the proposed method. 

The rest of this paper is organized as follows. In Section 2, we introduce the eTPR model and Bayesian estimation method, including the Bayesian inference and variable selection procedure. It also presents statistical properties. Numerical studies are given in Section 3. A few concluding remarks are given in  Section 4. All proofs  are listed in Supplementary Material.

\section{Methodology}
\subsection{eTPR model}

A random function $f$ is said to follow an ETP, $%
f\sim ETP(\nu ,\omega ,h,k),$  if for any collection of points $%
\mbox{\boldmath
${X}$}=(\mbox{{\boldmath${x}$}$_{1}$},...,\mbox{{\boldmath${x}$}$_{n}$}%
)^{T},\mbox{{\boldmath${x}$}$_{i}$}\in {\mathcal{X}} \subset R^p$, 
$
\mbox{{\boldmath${f}$}$_{n}$}=f(\mbox{{\boldmath ${X}$}$$})=(f(%
\mbox{{\boldmath${x}$}$_{1}$}),...,f(\mbox {{\boldmath${x}$}$_{n}$}%
))^{T}$ has an extended multivariate $t$ distribution (EMTD), \\$EMTD(\nu ,\omega ,\mbox{{\boldmath
${h}$}$_{n}$},\mbox{{\boldmath${K}$}$_{n}$}),
$
where the density function is
\begin{equation*}
p(z)=|2\pi \omega \mbox{{\boldmath ${K}$}$_{n}$}|^{-1/2}\frac{\Gamma
(n/2+\nu )}{\Gamma (\nu )}\left( 1+\frac{(z-\mbox{{\boldmath ${h}$}$_{n}$}%
)^{T}\mbox{{\boldmath ${K}$}$_{n}^{-1}$}(z-\mbox{{\boldmath ${h}$}$_{n}$})}{%
2\omega }\right) ^{-(n/2+\nu )},
\end{equation*}
$\mbox{{\boldmath${h}$}$_{n}$}=(h(\mbox{{\boldmath${x}$}$_{1}$}),...,h(%
\mbox{{\boldmath${x}$}$_{n}$}))^{T}$, $\mbox{{\boldmath
${K}$}$_{n}$}=(k_{ij})_{n\times n}$ with $k_{ij}=k(%
\mbox{{\boldmath
${x}$}$_{i}$},\mbox{{\boldmath${x}$}$_{j}$})$ for some mean function $%
h(\cdot ): {\mathcal{X}}\rightarrow R$ and kernel function $k(\cdot ,\cdot ):  {\mathcal{X}}\times
 {\mathcal{X}}\rightarrow R.$

Following \cite{Wang2017extend}, we assume that for model (\ref{prmodel}), $f_{i}$ and $\varepsilon_{i}$ have a joint extended t-process,
\begin{equation}\label{etpmodel}
\left( \begin{array}{c}
f_{i} \\
\varepsilon_{i}
\end{array}
\right) \sim ETP \left( \nu,\omega,\left( \begin{array}{c}
h_{i} \\
0
\end{array}
\right),\left( \begin{array}{cc}
k_{i} & 0 \\
0 & \widetilde{k}
\end{array}
\right) \right),
\end{equation}
where $h_{i}$ and $k_{i}$ are respectively mean and kernel functions, and $\widetilde{k}(\ve{u},\ve{v})=\sigma^{2} I(\ve{u}=\ve{v})$.
Let observed data set $\vesub{D}{n}=\{ (\vesub{X}{i},\vesub{y}{i}): i=1,\dots,m \}$, where $\vesub{y}{i}=(y_{i1},\dots,y_{in_{i}})^{\top}$ are the observed responses and
$\vesub{X}{i}=(\vesub{x}{i1},\dots,\vesub{x}{in_{i}})^{\top}$, $\vesub{x}{ij} \in \emph{\textbf{R}}^{p}$, are observed covariates.
Without loss of generality, let $n_{1}=\cdot\cdot\cdot=n_{m}=n$, and $h_i(\cdot)=0.$  It shows that model (\ref{etpmodel}) can be rewritten hierarchically as
\begin{equation}
\vesub{y}{i}|\vesub{f}{i},r_{i} \stackrel{ind}{\sim} N(\vesub{f}{i},r_{i}\sigma^{2}\vesub{I}{n}), ~~
\vesub{f}{i}|r_{i} \stackrel{ind}{\sim} N(\vesub{h}{in},r_{i}\vesub{K}{in}), ~~r_{i} \stackrel{ind}{\sim} IG(\nu,\omega), \label{EMTD}
\end{equation}
where $\vesub{f}{i}=f_{i}(\vesub{X}{i})$,  $\vesub{h}{in}=h_{i}(\vesub{X}{i})$, $\vesub{K}{in}=(k_{ijl})_{n \times n}$ with $k_{ijl}=k_{i}(\vesub{x}{ij},\vesub{x}{il})$, and
$N$ and $IG$ stand for a normal distribution and an inverse gamma distribution respectively.
From \cite{Wang2017extend}, we set $\omega=\nu-1$. The parameter $\nu$ can be treat as degree of freedom for the eTPR model.

\subsection{Estimation procedure}

To estimate $f_i$, we first need to estimate the unknown parameters involved in the covariance 
function $k_{i}(\cdot,\cdot)$. A function family such as a squared exponential kernel and Mat\'{e}rn class kernel can
be applied \citep[see e.g.][]{Shi2011Gaussian}. This paper takes a combination of a square exponential kernel and a non-stationary linear kernel,
\begin{align}\label{kf}
k_{i}(\vesub{x}{ij},\vesub{x}{il})
&=k(\vesub{x}{ij},\vesub{x}{il};\vesub{\beta}{i})  \notag \\
&=v_{i}\exp(-\frac{1}{2}\sum_{q=1}^{p}w_{i,q}(x_{ij,q}-x_{il,q})^{2})+
\sum_{q=1}^{p}a_{i,q}x_{ij,q}x_{il,q},
\end{align}
where $\vesub{\beta}{i}=(v_{i},w_{i,1},\dots,w_{i,p},a_{i,1},\dots,a_{i,p}), i=1,...,m,$ is a vector of hyper-parameters.

From model (\ref{EMTD}), we have a joint likelihood function
\begin{align}\label{jll}
L(\ve{Y}|\nu,\sigma^{2},\vesub{\beta}{1},\dots,\vesub{\beta}{m})=
\prod_{i=1}^{m}L_{i}(\vesub{y}{i}|\nu,\sigma^{2},\vesub{\beta}{i}),
\end{align}
where $\ve{Y}=(\vess{y}{1}{\top},\dots,\vess{y}{m}{\top})^{\top}$, and $L_{i}$ is the likelihood function based on the data observed from the $i$-th subject.
Maximizing (\ref{jll}) over $\ve\theta=(\vesub{\beta}{1},...,\vesub{\beta}{m},\sigma^2)^\top$ and $\nu$,  \cite{Wang2017extend} obtain the MLEs of $\ve\theta$ and $\nu$.
However, $\nu$ is usually over-estimated using the likelihood method as we discussed  in the previous section, and then, it may lose the robustness. 

Instead of using MLE, this paper applies Bayesian method to estimate the unknown parameters.	
For $\sigma^{2}$ and $\vesub{\beta}{i}$,  we take the following hyper-prior distributions,
\begin{equation}\label{pd}
\begin{aligned}
&w_{i,q}^{-1} \stackrel{ind}{\sim} G(\alpha_{1},\mu_{1}),~\log a_{i,q} \stackrel{ind}{\sim} N(\mu_{2},\sigma_{2}^{2}),\\
&\log v_{i} \stackrel{ind}{\sim} N(\mu_{3},\sigma_{3}^{2}),~ \log \sigma^{2} \sim N(\mu_{4},\sigma_{4}^{2}),
\end{aligned}
\end{equation}
where $G(\alpha_{1},\mu_{1})$ stands for a gamma distribution  with parameters $\alpha_{1}$ and $\mu_{1}$.
We also specify a prior for $\nu$: 
\begin{equation}
\pi(\nu)d\nu \propto \nu^{-2}d\nu, (\nu \geq 1).\label{prior-nu}
\end{equation}

By combining the likelihood function (\ref{jll}) and the prior densities, we have
a joint posterior likelihood function of the parameters, 
\begin{align}
\pi(\nu,\sigma^{2},\vesub{\beta}{1},\dots,\vesub{\beta}{m}|\ve{Y}) \propto
\prod_{i=1}^{m}L_{i}(\vesub{y}{i}|\nu,\sigma^{2},\vesub{\beta}{i}) \pi(\nu)\pi(\sigma^{2})\pi(\vesub{\beta}{1})\cdot\cdot\cdot\pi(\vesub{\beta}{m}),\nonumber
\end{align}
where $\pi(\nu)$, $\pi(\sigma^2)$ and $\pi(\vesub{\beta}{i})$ are the density functions defined in (\ref{pd}) and (\ref{prior-nu}) for the priors of $\nu$, $\sigma^2$ and $\vesub{\beta}{i}$, respectively.
Let
$l(\nu,\ve{\theta};\ve{Y})=\log (\pi(\nu,\sigma^{2},\vesub{\beta}{1},\dots,\vesub{\beta}{m}|\ve{Y}))$.
The parameters are estimated by maximizing $l(\nu,\ve{\theta};\ve{Y})$ over $\ve{\theta}$ and $\nu$. 

Note that the kernel function (\ref{kf}) includes $m(2p+1)$ hyper-parameters. With large $p$, there are too many parameters.
This paper develops a spike and slab variable selection method \citep{Ishwaran2005Spike, Yen2011A} for model (\ref{prmodel}).
Applying 
the spike and slab priors in (\ref{pd}), we define new hyper-prior distributions as follows. 
\begin{equation}\label{pd-new}
\begin{aligned}
&w_{i,q}|\gamma_{i,q} \stackrel{ind}{\sim} \gamma_{i,q}IG(\alpha_{1},\mu_{1}^{-1})+(1-\gamma_{i,q})I(w_{i,q}=0),\\
&\gamma_{i,q}|\kappa \stackrel{ind}{\sim} Bernoulli(\kappa),~q=1,\dots,p,\\
&a_{i,q}|\delta_{i,q} \stackrel{ind}{\sim} \delta_{i,q}LogN(\mu_{2},\sigma_{2}^{2})+(1-\delta_{i,q})I(a_{i,q}=0),\\
&\delta_{i,q}|\kappa \stackrel{ind}{\sim} Bernoulli(\kappa),~q=1,\dots,p,
\end{aligned}
\end{equation}
where $Bernoulli$ and $LogN$ stand for a Bernoulli distribution and a log-normal distribution, respectively.
Maximizing the posterior likelihood with the spike and slab priors (\ref{pd-new}),
we force some parameters involved in \eqref{kf} to be zero (resulting in a simpler covariance structure), and in the meantime, we  
obtain the estimates of the selected (non-zero) parameters. 

\subsection{Prediction and consistency}

At a new observed point $\ve{u}$, we show that
\begin{equation}
f_{i}(\ve{u})|\vesub{D}{n} \sim EMTD(n/2+\nu,n/2+\nu-1,\mu_{in}^{*},\sigma_{in}^{*}),
\end{equation}
where $\mu_{in}^{*}=E(f_{i}(\ve{u})|\vesub{D}{n})=\vess{k}{iu}{\top}\tvess{\Sigma}{in}{-1} \vesub{y}{i}$,  $\sigma_{in}^{*}=Var(f_{i}(\ve{u})|\vesub{D}{n})=s_{0i}(k_{i}(\ve{u},\ve{u})-\vess{k}{iu}{\top}\tvess{\Sigma}{in}{-1}\vesub{k}{iu}),$
$s_{0i}=E(r_{i}|\vesub{D}{n})=(\vess{y}{i}{\top}\widetilde{\boldsymbol{\Sigma}}_{in}^{-1}\vesub{y}{i}+2(\nu-1))/(n+2(\nu-1))$,
$\vesub{k}{iu}=(k_{i}(\vesub{x}{i1},\ve{u}),\dots,k_{i}(\vesub{x}{in},\ve{u}))^{\top}$, and $\tvesub{\Sigma}{in}=\sigma^{2}\vesub{I}{n}+\vesub{K}{in}$.
 By replacing the unknown parameters in $\mu_{in}^{*}$ and $\sigma_{in}^{*}$ with their estimates, 
 it gives a prediction of ${y}_{i}(\ve{u})$, denoted by $\hat{y}_{i}(\ve{u})=\mu_{in}^{*}$, and an estimate of its variance.


Let $P(\vesub{y}{i}|f_{i},\vesub{X}{i})$ be density function of $\vesub{y}{i}$ with function $f_{i}$ under eTPR,
and $P_0(\vesub{y}{i}|\vesub{X}{i})=P(\vesub{y}{i}|f_{0i},\vesub{X}{i})$ where $f_{0i}$ is the true underlying function of $f_{i}$.
Let $P_{bs}(\vesub{y}{i}|\vesub{X}{i})$ represent a Bayesian TP prediction strategy with
$P_{bs}(\mbox{{\boldmath ${y}$}$_{i}$}|
\mbox{{\boldmath
${X}$}$_{i}$})=\int_{{\mathcal{F}}}P(\mbox{{\boldmath
${y}$}$_{i}$}|f,\mbox{{\boldmath
${X}$}$_{i}$})dp_{\beta_i}(f)$,
where  $p_{\beta_i}(f)$ is an ETP prior with the kernel function $k_i$ (a measure of random
process $f$ on space ${\mathcal{F}}=\{f(\cdot):\mathcal{X}\rightarrow R\}$ deduced by kernel function $k_i$).
Then we have the following theorem. 

\newtheorem{theorem}{Theorem}
\begin{theorem}
Suppose $\vesub{y}{i}=(y_{i1},\dots,y_{in})$ are generated from the eTPR model (\ref{EMTD}) with  the covariance kernel function $k_{i}$. Let $k_{i}$ be bounded and thrice differentiable in parameter $\vesub{\beta}{i}$. Then we have\\
(i) The MAP estimator $\hvesub{\beta}{i}$ 
is a consistent estimator of $\vesub{\beta}{i}$. \\
(ii) Prediction strategy has information consistency,
\begin{equation}\label{incon}
\frac{1}{n}E_{\vesub{X}{i}}(D[P_0(\vesub{y}{i}|\vesub{X}{i}) \|P_{bs}(\vesub{y}{i}|\vesub{X}{i})])\longrightarrow0, as \ n\rightarrow \infty,
\end{equation}
where the expectation is taken over the distribution of $\vesub{X}{i}$, and $D[P_{1}\|P_{2}]=\int (\log P_{1}-\log P_{2})d P_{1}$ denotes the Kullback-Leibler divergence.
\end{theorem}

The proof is given in Supplementary Material.

For GPR and eTPR, \cite{Seeger2008Information} and  \cite{Wang2017extend}  studied information consistency of their proposed methods, respectively.
Theorem 1 shows information consistency under the proposed Bayesian estimation for eTPR (BeTPR).

\section{Numerical study}

\subsection{Simulation studies}

Predictions from BeTPR  are compared with those from
GPR and eTPR by simulation studies.  For priors of the parameters in (\ref{pd}), we take $\log a_{i,q} \sim N(-3,3^{2})$, $\log \sigma^{2} \sim N(-3,3^{2})$, $\log v_{i,1} \sim N(-3,1)$ and $w_{i,q}^{-1} \sim \Gamma(2,0.5)$. The parameter $\kappa$ for Bernoulli distribution in (\ref{pd-new}) takes 0.84.  More discussion on choosing or estimating $\kappa$ can be founded in \cite{Yen2011A}. All simulation results are based on 500 replications.


\begin{table}[h!]
	\centering
	\tabcolsep=2pt\fontsize{8}{12}\selectfont
	\caption{Mean squared errors  and the standard deviation (in parentheses) from the GPR, eTPR and BeTPR methods with $m=2$ and 5.}
\vskip 10pt
	\begin{tabular}{cccccccc}
		\hline
		\multirow{2}{*}{Case} & \multicolumn{3}{c}{$m=2$} & &  \multicolumn{3}{c}{$m=5$}  \\
		\cline{2-4} \cline{6-8}
		& GPR & eTPR & BeTPR &  & GPR & eTPR & BeTPR \\
		\hline
		(1) & 0.141(0.399) & 0.127(0.323)   & 0.081(0.279)  & & 0.221(0.426)  & 0.171(0.399) &  0.104(0.310)\\
		(2) & 0.118(0.358) & 0.107(0.301)   & 0.058(0.150) & & 0.123(0.204) &  0.097(0.162)  & 0.048(0.066)\\
		(3) & 0.175(0.410) & 0.158(0.386)  & 0.092(0.278)  & & 0.229(0.447) &  0.198(0.452)  & 0.098(0.292)\\
		(4) & 0.187(0.493) & 0.152(0.357)  & 0.086(0.210)  & & 0.203(0.350) &  0.160(0.343)  & 0.091(0.163)\\
		\hline
	\end{tabular}
	\label{tab1}
\end{table}

\begin{table}[h!]
	\centering
	\tabcolsep=2pt\fontsize{8}{12}\selectfont
	\caption{Estimates of $\nu$ and their standard deviation (in parentheses) from the eTPR and BeRPR methods with $m=2$ and 5.}
\vskip 10pt
	\begin{tabular}{cccccc}
		\hline
		\multirow{2}{*}{Case} & \multicolumn{2}{c}{$m=2$}  & &  \multicolumn{2}{c}{$m=5$}  \\
		\cline{2-3} \cline{5-6}
		& eTPR & BeTPR  & & eTPR & BeTPR \\
		\hline
		(1) & 2.587(0.769) & 1.148(0.109)  & & 2.309(0.861)  & 1.180(0.310)  \\
		(2) & 2.648(0.720) &  1.149(0.096) & & 2.545(0.748) &  1.140(0.103)   \\
		(3) & 2.620(0.738) &  1.186(0.152)  & & 2.260(0.862) &  1.209(0.253)  \\
		(4) & 2.571(0.770) &  1.178(0.134)  & & 2.229(0.884) &  1.202(0.263)  \\
		\hline
	\end{tabular}
	\label{tab2}
\end{table}

Simulated data with $p=1$ are generated from the following 4 cases: \\
(1) $f_{i} \sim GP(0,k_{i})$, $\varepsilon_{i} \sim N(0,\sigma^{2})$, $\sigma^{2}=0.05$, and $\vesub{\beta}{i}=\vesub{\beta}{0}$; \\
(2) $f_{i} \sim GP(0,k_{i})$, $\varepsilon_{i} \sim \sigma t_{2}$, $\sigma^{2}=0.05$, and $\vesub{\beta}{i}=\vesub{\beta}{0}$; \\
(3) $f_{i} \sim ETP(2,2,0,k_{i})$, $\varepsilon_{i} \sim ETP(2,2,0,\widetilde{k})$, $\sigma^{2}=0.05$, and $\vesub{\beta}{i}=\vesub{\beta}{0}$;  \\
(4) $f_{i}$ and $\varepsilon_{i}$ have a joint ETP with $\sigma^{2}=0.05$ and $\vesub{\beta}{i}=\vesub{\beta}{0}$; \\
where $\vesub{\beta}{i}=(v_{i},w_{i,1},a_{i,1})$ are hyper-parameters in $k_{i}$, $i=1,\dots,m$, and $\vesub{\beta}{0}=(0.025,2,0.025)$. For each covariates, $N = 50$ points are generated evenly spaced in $[0,3]$, and $n=10$ points are randomly selected as training data and the remaining as test data. Besides, to study robustness, in Cases (1), (3) and (4), one sample is randomly selected from the training data and is added with an extra  error generated from $t_{2}$ (t-distribution with DF of 2). Table \ref{tab1} presents mean squared errors (MSE) between the test data and the prediction from GPR, eTPR and BeTPR and the standard deviation of the prediction, where $m=2$ and 5.
It shows that BeTPR has the smallest MSEs, while eTPR does perform better than GPR which is consistent with the findings in  \cite{Wang2017extend}.
Table \ref{tab2} shows the estimates of $\nu$ from eTPR and BeTPR. We see that BeTPR has much smaller estimates of $\nu$ than eTPR, which indicates that
BeTPR performs more robust than eTPR.


\begin{table}[h!]
	\centering
\tabcolsep=4pt\fontsize{8}{12}\selectfont
	\caption{Mean squared errors  and the standard deviation (in parentheses) from the  eTPR, BeTPR and BeTPR(Variable Selection) methods with $p=3$, $m=2$ and 5.}
\vskip 10pt
	\begin{tabular}{ccccc}
		\hline
		$m$ & Case & eTPR & BeTPR & BeTPR(VS)  \\
		\hline
		2 & (5)  & 0.103(0.102) & 0.073(0.055)  & 0.065(0.048) \\
		& (6)  & 0.176(0.175)   & 0.163(0.181) & 0.136(0.148) \\
		\cline{2-5}
		5 & (5)   & 0.111(0.092) & 0.069(0.045) &  0.068(0.045) \\
		& (6)  &  0.160(0.105)  & 0.113(0.081) & 0.109(0.069) \\
		\hline
	\end{tabular}
	\label{tab5}
\end{table}


We also investigate performance of variable selection of the BeTPR method (BeTPR(VS)) by simulation studies with $p=3$.
Data are generated from models:\\
(5) $f_{i} \sim GP(0,k_{i})$, $\varepsilon_{i} \sim N(0,\sigma^{2})$, $\sigma^{2}=0.05$, and $\vesub{\beta}{i}=(0.5,1,0,0,0.5,0,0)$; \\
(6) $f_{i} \sim ETP(2,2,0,k_{i})$, $\varepsilon_{i} \sim ETP(2,2,0,\widetilde{k})$, $\sigma^{2}=0.05$, and $\vesub{\beta}{i}=$ $(0.5,1,0,$ $0,0.5,0,0)$. \\
As before,  the first covariate takes $N = 50$ points which are evenly spaced in $[5,10]$;
and for the other two covariates,  they are generated from $N(0,0.1)$.
Simulation results show that the mean accuracies of variable selection are $91.9\%$ and $94.2\%$ for the square exponential kernel and non-stationary kernel, respectively.
Table \ref{tab5} presents prediction results from eTPR, BeTPR and BeTPR(VS).
We find that the BeTPR(VS) has the smallest prediction errors, which shows that the Bayesian method including variable selection can improve the performance further.

\subsection{Real examples}

\begin{table}[h!]
	\centering
\tabcolsep=4pt\fontsize{8}{12}\selectfont
	\caption{Prediction errors and the standard deviation (in parentheses) from the GPR, eTPR and BeTPR methods for an executive function research data and market penetration data.}
	\begin{tabular}{ccccc}
		\hline
		Data & $m$ & GPR  & eTPR & BeTPR  \\
		\hline
		DMS  & 2 & 0.271(0.048)   & 0.239(0.033)   & 0.229(0.027) \\
		SWM  & 2 & 0.082(0.024)  &  0.068(0.012)  &  0.065(0.012)  \\
		TD-Australia   & 2 & 0.057(0.041)   & 0.053(0.053)   & 0.038(0.034)  \\
		TD-Asia Pacific   & 3 & 0.010(0.018) &  0.004(0.006) &  0.003(0.003)  \\
		WM-Australia   & 2 & 0.081(0.083)  &  0.040(0.040) &   0.024(0.022)  \\
		WM-Asia Pacific   & 9 & 0.092(0.039)  & 0.078(0.036)  &  0.050(0.024)  \\
		\hline
	\end{tabular}
    \label{tab6}
\end{table}

The BeTPR method is applied to two datasets: an executive function research data and market penetration of new product data.
The executive function research data comes from the study in children with Hemiplegic Cerebral Palsy. The data set consists of 84 girls and 57 boys from primary and secondary schools, which were subdivided into two groups ($m=2$): the action video game players group (AVGPs)($56\%$) and the non action video game players group (NAVGPs)($44\%$). In this paper,  we select two measurement indices: mean token search preparation time (SWM) and  mean correct latency (DMS); the details can be found in \cite{Xu2015Automatic} and \cite{Wang2017extend}. The market  data contains market penetrations of 760 categories drawn from 21 new products and 70 countries; see the details in \cite{Sood2009Functional}. In this paper, we take penetration data of Tumble Drier~(TD) and Washing Machine~(WM) from  1977 to 2015 in two regions: Australia and Asia Pacific.
The countries with  positive penetration in the beginning year of 1977 (non null or non zero) are selected, such that $m$ for TD in these two regions are 2 and 3, and those for WM are
2 and 9 respectively.

To measure the performance, we randomly select $60\%$ observations as the training data and the remaining as the test data.
Three methods are applied to fit the training data and to predict the test data. This procedure is repeated 500 times. Table \ref{tab6} presents mean prediction errors from GPR, eTPR and BeTPR. As we expected, BeTPR has the best performance, especially for market penetration data which include many outliers as the nature of such data. This shows that BeTPR provides a robust method. 

\section{Conclusions}

This paper uses a Bayesian method to  estimate the parameters involved in  the eTPR model.
Compared with the MLE method, the proposed method can avoid an over-estimation of the DF $\nu$, and thus provide a stable robust method in the presence of outliers.
 Statistical properties, including consistency of the parameter estimation and information consistency  are obtained.
This paper assumes that prior of the unknown function and the error term have a joint ETP, which is an unnatural way to define a process model \citep[see discussions in][]{Wang18}. A better way is to use an independent processes model, i.e. the prior of the unknown function and the error terms are independent.  But this model makes
estimation procedure more complicated because of the involvement of intractable multidimensional integrations. We leave the issue for  future research.


\bibliographystyle{elsarticle-harv}
\bibliography{bib}

\end{document}


\begin{center}
{\Large Supplementary: A robust estimation for the extended t-process functional regression models}
\end{center}
\begin{center}		
{Zhanfeng Wang},   {Kai Li}  and {Jian Qing Shi}
\end{center}
\begin{center}
{Department of Statistics and Finance, Management School, University of Science and Technology of China, Hefei, China.}\\
{School of Mathematics and Statistics, Newcastle University, Newcastle, UK.}
\end{center}


\hskip 1cm

For convenience, we define slightly different notations. Let $\vesup{Y}{n}=(Y_{1},\dots,Y_{n})$ for $n\geq 1$ be a sequence of random variables 
with density function $p(\vesup{y}{n};\ve{\theta})=p(y_{1},\dots,y_{n};\ve{\theta})$. Let $\vesub{\theta}{0}$ be the true value of $\ve{\theta}$ and for every $k\geq 1$,
$$p_{k}(\ve{\theta})=p(\vesup{y}{k};\ve{\theta})/p(\vesup{y}{k-1};\ve{\theta}).$$
This paper assumes that $p_{k}(\ve{\theta})$ is twice differentiable with respect to $\ve{\theta}$ and the support of $p(\vesup{y}{n};\ve{\theta})$ is independent of $\ve{\theta}$. Define $\phi_{k}(\ve{\theta})=\log p_{k}(\ve{\theta})$, and let $\vesub{U}{k}(\ve{\theta})$ and $\vesub{V}{k}(\ve{\theta})$ be the first and second derivative matrix of $\phi_{k}(\ve{\theta})$ with respect to $\ve{\theta}$, respectively.

Without loss of generality, we consider one-dimensional parameter $\theta$ and its true value $\theta_{0}$. Hence, $\vesub{U}{k}(\ve{\theta})$ and $\vesub{V}{k}(\ve{\theta})$ become scalers $U_{k}(\theta)$ and $V_{k}(\theta)$, saying $U_{k}=U_{k}(\theta_{0})$ and $V_{k}=V_{k}(\theta_{0})$. Then, for the consistency of maximum likelihood estimators, we list the following conditions \cite{Rao1980Statistical}:  \\
(C1) $\phi_{k}(\theta)$ is thrice differentiable for all $\theta$ in a compact $\Theta$, denoted by $W_{k}(\theta)$ the third derivative of $\phi_{k}(\theta)$. \\
(C2) Twice Differentiation of $p(\vesup{y}{n};\theta)$ with respect to $\theta$ is permitted under the integral sign with $\int p(\vesup{y}{n};\theta)d\mu^{n}(\vesup{y}{n})$.\\ 
(C3) $E|V_{k}|<\infty$, $E|Z_{k}|<\infty$ where $Z_{k}=V_{k}+U_{k}^{2}$. \\
(C4) Let $i_{k}(\theta_{0})=var[U_{k}|\mathscr{F}_{k-1}]=E[U_{k}^{2}|\mathscr{F}_{k-1}]$, $I_{n}(\theta_{0})=\sum_{k=1}^{n}i_{k}(\theta_{0})$, $S_{n}=\sum_{k=1}^{n}U_{k}$ and $S_{n}^{*}=\sum_{k=1}^{n}V_{k}+I_{n}(\theta_{0})$. There exists a sequence of constants $K(n) \rightarrow \infty$ as $n \rightarrow \infty$ such that
$\{K(n)\}^{-1}S_{n}\stackrel{p}{\rightarrow}0,$ $\{K(n)\}^{-1}S_{n}^{*}\stackrel{p}{\rightarrow}0,$ and
there exists $a(\theta_{0})>0$ such that for every $\varepsilon>0$, $P[\{K(n)\}^{-1}I_{n}(\theta_{0})\geq 2a(\theta_{0})]\geq 1-\varepsilon$ for all $n\geq N(\varepsilon)$,
 $\{K(n)\}^{-1}\sum_{k=1}^{n}E|W_{k}(\theta)|<M<\infty$ for all $\theta \in I$ and for all $n$.

Moreover, for the priors of the parameter, we need\\
(C5) $l(\theta)=\log \pi(\theta)$ is twice differentiable with respect to $\theta$ for all $\theta \in \Theta$.

\newtheorem{lemma}{Lemma}
\begin{lemma}
Under the conditions (C1)-(C5), the maximum a posterior(MAP) estimator of $\theta$ is consistent for $\theta_{0}$ as $n \rightarrow \infty$.
\end{lemma}
\begin{proof}
	The posterior density function of $\theta$ can be obtained by using Bayes's theorem as $\pi(\theta;\vesup{y}{n})=c\Pi_{k=1}^{n}p_{k}(\theta)\pi(\theta)$, where $c$ is the regularization constant. Then from (C1) and (C5), we show the following Taylor's expansion
	\begin{equation}\label{A1}
	\frac{\partial \log \pi(\theta;\vesup{y}{n})}{\partial \theta}=\sum_{k=1}^{n}U_{k}+(\theta-\theta_{0})\sum_{k=1}^{n}V_{k}(\theta_{n}^{*})+l'(\theta_{0})+(\theta-\theta_{0})l''(\theta_{n}^{*}),
	\end{equation}
	where $\theta_{n}^{*}=\theta_{0}+r(\theta-\theta_{0})$ with $r=r(n,\theta_{0})$ satisfying $|r|\leqq 1$.
	
	Following the proof of \cite{Rao1980Statistical}, (C1)-(C4) imply that
	\begin{equation}\label{A2}
	\frac{1}{K(n)}\sum_{k=1}^{n}U_{k}\stackrel{p}{\rightarrow}0,
	\end{equation}
	and for any $\varepsilon>0$, there exist $\eta>0$ and $N(\varepsilon)$ such that for $n>N(\varepsilon)$,
	\begin{equation}\label{A3}
	P[\{K(n)\}^{-1}|\sum_{k=1}^{n}V_{k}|\geq a(\theta_{0})]> 1-\varepsilon,
	\end{equation}
	\begin{equation}\label{A4}
	P[\{K(n)\}^{-1}|\sum_{k=1}^{n}(V_{k}(\theta_{n}^{*})-V_{k})|> (1+\eta)^{-1}a(\theta_{0})] <\varepsilon.
	\end{equation}
	In addition, from (C5) we have
	\begin{equation}\label{A5}
	\frac{1}{K(n)}( l'(\theta_{0})+(\theta-\theta_{0})l''(\theta_{n}^{*}) ) \rightarrow 0.
	\end{equation}
	Notice that (\ref{A5}) can be absorbed in (\ref{A2}) by Slutsky's Theorem,  following \cite{Rao1980Statistical}, (\ref{A2})-(\ref{A5}) lead to the result that (\ref{A1}) has a root $\hat{\theta}$ which is consistent for $\theta_{0}$ as $n \rightarrow \infty$. 	
\end{proof}

\begin{proof}[Proof of Theorem 1]
	
	
	With abuse of notation, let $\vesup{Y}{n}=(Y_{1},\dots,Y_{n})$ with $Y_{i}=\vesub{y}{i}$. Notice that $\vesup{Y}{n}$ has an extended
	multivariate t-distribution with mean $\ve{0}$ and covariance $\tvesub{\Sigma}{in}$. 
	And for given $k$, the covariance matrix of $\vesup{Y}{k}$ has the following partition
	$$ \tvesub{\Sigma}{ik}=\left( \begin{array}{cc}
	Cov(\vesup{Y}{k-1},\vesup{Y}{k-1}) & Cov(\vesup{Y}{k-1},Y_{k}) \\
	Cov(Y_{k},\vesup{Y}{k-1}) &  Cov(Y_{k},Y_{k})
	\end{array}
	\right)\triangleq \left( \begin{array}{cc}
	\tvesub{\Sigma}{k-1,k-1} & \tvesub{\Sigma}{k-1,1} \\
	\tvesub{\Sigma}{1,k-1} &  \tvesub{\Sigma}{1,1}
	\end{array}
	\right). $$
From \cite{Wang2017extend}, $Y_{k}$ conditional on $\vesup{Y}{k-1}=\vesup{y}{k-1}$ is also an extended t-process with mean $m_{k}(\vesub{\beta}{i})$, covariance $v_{k}(\vesub{\beta}{i})$ and degree of freedom $\nu^{*}$, where
	$$m_{k}(\vesub{\beta}{i})=\tvesub{\Sigma}{1,k-1}\tvesub{\Sigma}{k-1,k-1}^{-1}\vesup{y}{k-1},$$	 $$v_{k}(\vesub{\beta}{i})=\frac{2\omega+(\vesup{y}{k-1})^{T}\tvesub{\Sigma}{k-1,k-1}^{-1}\vesup{y}{k-1}}{2\omega+k-1}(\tvesub{\Sigma}{1,1}-\tvesub{\Sigma}{1,k-1}\tvesub{\Sigma}{k-1,k-1}^{-1}\tvesub{\Sigma}{k-1,1}),$$
	$$\nu^{*}=\nu+(k-1)/2.$$ 
	 According to Lemma 1, we only need to verify conditions (C1)-(C5) for $\vesup{Y}{n}$. Thus, without loss of generality, assuming that $\vesub{\beta}{i}$ is a scalar $\theta$ with the true value $\theta_{0}$, $\phi_{k}(\theta)$ and its derivatives can be given by
	\begin{align}\label{A6}
	& \phi_{k}(\theta)= -\log(\sqrt{2\pi\omega^{*} v_{k}(\theta)})-(\frac{1}{2}+\nu^{*})\log (1+\frac{z_{k}^{2}}{2\omega^{*}})+C,  \notag\\
	& U_{k}(\theta)=-\frac{v_{k}^{'}(\theta)}{2v_{k}(\theta)}+(\frac{1}{2}+\nu^{*})\frac{2\omega^{*}}{2\omega^{*}+z_{k}^{2}}\{ \frac{m_{k}^{'}(\theta)z_{k}}{\omega^{*}\sqrt{v_{k}(\theta)}}+\frac{v_{k}^{'}(\theta)z_{k}^{2}}{2\omega^{*} v_{k}(\theta)} \},  \notag\\
	& V_{k}(\theta)=\frac{A_{k}(\theta)z_{k}^{4}+B_{k}(\theta)z_{k}^{3}+C_{k}(\theta)z_{k}^{2}+D_{k}(\theta)z_{k}+E_{k}(\theta)}{(2\omega^{*}+z_{k}^{2})^{2}}+F_{k}(\theta),
	\end{align}
	where $z_{k}=(y_{k}-m_{k}(\theta))/\sqrt{v_{k}(\theta)}|_{\theta=\theta_{0}}$ has an extended
		multivariate t-distribution with mean 0, covariance 1 and degrees of freedom $\nu^{*}$, $C$ is a constant indepedent of $\theta$, and $A_{k}(\theta),\dots,F_{k}(\theta)$ consist of $m_{k}(\theta)$, $v_{k}(\theta)$, and their first and second derivatives.
	
	Since $\vesup{Y}{k}$ has an extended
	multivariate t-distribution and the covariance kernel function $k_{i}$ is thrice differentiable, it easily shows that (C1) and (C2) hold for the eTPR models.
	Besides, (C5) obviously holds when the priors are chosen in the form of log-normal and inverse gamma.
	So we only need to verify (C3) and (C4).
	
	For (C4), let $K(n)=I_{n}(\theta_{0})=\sum_{k=1}^{n}i_{k}(\theta_{0})$ with $i_{k}(\theta_{0})=E[U_{k}^{2}|\mathscr{F}_{k-1}]$. Then obviously, (\romannumeral3) in (C4) directly holds. From $U_{k}(\theta)$ in (\ref{A6}), we show that $$i_{k}(\theta_{0})=E[U_{k}^{2}|\mathscr{F}_{k-1}]=\frac{(1+2\nu^{*})\nu^{*}}{(3+2\nu^{*})\omega^{*}}\frac{m_{k}^{'}(\theta_{0})^{2}}{v_{k}(\theta_{0})}+\frac{\nu^{*}}{3+2\nu^{*}} \left( \frac{v_{k}^{'}(\theta_{0})}{v_{k}(\theta_{0})} \right)^{2}.$$
    To prove the remaining conditions, we shall point out that $i_{k}(\theta_{0})=O(1)$, i.e. there exists the constants $m^{'}>0$ and $M^{'}>0$ independent of $\mathscr{F}_{k-1}$ such that
    \begin{align}\label{A7}
     m^{'} \leq i_{k}(\theta_{0}) \leq M^{'}.
    \end{align}
    The right inequality of (\ref{A7}) can be easily obtained with the fact that $m_{k}^{'}(\theta_{0})^{2}$, $v_{k}^{'}(\theta_{0})$ and $v_{k}(\theta_{0})$ are all $O((\vesup{y}{k-1})^{T}\vesup{y}{k-1})$. And as for the left inequality, if $i_{k}(\theta_{0})=o(1)$, then we have
    \begin{align}\label{A8}
    \frac{m_{k}^{'}(\theta_{0})^{2}}{v_{k}(\theta_{0})}=o(1), \left( \frac{v_{k}^{'}(\theta_{0})}{v_{k}(\theta_{0})} \right)^{2}=\left( \frac{s_{k}^{'}}{s_{k}}+\frac{w_{k}^{'}}{w_{k}}\right)^{2}=o(1),
    \end{align}
    where $s_k=(2\omega+(\vesup{y}{k-1})^{\top}\tvesub{\Sigma}{k-1,k-1}^{-1}\vesup{y}{k-1})/(2\omega+k-1)$ and $w_{k}=\tvesub{\Sigma}{1,1}-\tvesub{\Sigma}{1,k-1}\tvesub{\Sigma}{k-1,k-1}^{-1}\tvesub{\Sigma}{k-1,1}$. Notice that $w_{k}^{'}/w_{k}$ is a constant independent of $\mathscr{F}_{k-1}$, and $m_{k}^{'}(\theta_{0})^{2}/v_{k}(\theta_{0})=o(1)$ implies that $s_{k}^{'}/s_{k}=o(1)$, so we have
    $$\left( \frac{v_{k}^{'}(\theta_{0})}{v_{k}(\theta_{0})} \right)^{2}=\left( \frac{s_{k}^{'}}{s_{k}}+\frac{w_{k}^{'}}{w_{k}}\right)^{2}=\left(\frac{w_{k}^{'}}{w_{k}}\right)^{2}+o(1)=O(1),$$
    which is opposite to (\ref{A8}). Hence, $m^{'}>0$.

    From (\ref{A7}), we have $K(n)=O(n)$. Similar to Appendix A.6 in \cite{Shi2011Gaussian}, it is sufficient to prove the remaining conditions in (C4), details for the sufficiency can be founded in \cite{Hall1980Martingale}.
    Moreover, (\ref{A7}) also implies that
    $$E(U_{k}^{2})=E(E[U_{k}^{2}|\mathscr{F}_{k-1}])\leq M^{'} <\infty.$$
    And similarly, we can get $E|V_{k}|<\infty$, thus (C3) holds.
    In conclusion, the maximum a posterior(MAP) estimator $\hvesub{\beta}{i}$ is consistent for $\vesub{\beta}{i0}$.

    For the second part of this theorem, the information consistency, its proof is similar to \cite{Wang2017extend}, so it is omitted here.

\end{proof}

\bibliographystyle{elsarticle-harv}
\bibliography{bib}